\documentclass{llncs}
\usepackage[T1]{fontenc}
\usepackage{cite}
\usepackage{amsmath,amssymb,amsfonts}
\usepackage{latexsym}
\usepackage[ruled,vlined]{algorithm2e}
\def\qed{\hfill$\Box$\par}
\SetCommentSty{small}
\SetAlgoVlined
\DontPrintSemicolon
\SetKw{KwLVar}{Local variable}
\SetKw{KwPVar}{Private local variable}
\SetKw{KwConst}{Constant}
\SetKw{KwMac}{Macro}
\SetKw{KwRules}{Rules}
\SetKwProg{Fn}{def}{\string:}{}
\SetKwProg{FnCont}{}{}{}
\SetKwIF{If}{ElseIf}{Else}{if}{:}{elif}{else:}{}%
\SetKwIF{Gc}{ElseGc}{ElseGcEnd}{}{~$\longrightarrow$}{$\Box$}{$\Box$}{}%
\SetKwFor{While}{while}{:}{fintq}%
\SetKwFor{For}{for}{\string:}{}%
\SetKwFor{Phase}{Phase}{:}{}%
\SetStartEndCondition{ }{}{}%
\SetNlSkip{0.5em}
\SetNlSty{}{}{:}
\SetKwFor{ForEach}{foreach}{\string:}{}%
\SetKwFunction{Range}{range}
\SetKw{Break}{break}

\def\XNULL{\bot}
\def\XBCAST{\textbf{bcast}}
\def\XRECV{\textbf{receive}}
\def\XNAPIERNUM{\mathrm{e}}
\begin{document}
\title{
  The R(1)W(1) Communication Model for Self-Stabilizing Distributed Algorithms
}
\author{ 
    Hirotsugu Kakugawa\inst{1}  \and
          Sayaka Kamei\inst{2}  \and \\
      Masahiro Shibata\inst{3}  \and
      Fukuhito Ooshita\inst{4}
}
\institute{
    Ryukoku University, Otsu, Shiga, Japan \\
    \email{kakugawa@rins.ryukoku.ac.jp}   \and
    Hiroshima University, Higashi Hiroshima, Hiroshima, Japan \\
    \email{s10kamei@hiroshima-u.ac.jp}    \and
    Kyushu Institute of Technology, Iizuka, Fukuoka, Japan \\
    \email{shibata@csn.kyutech.ac.jp}     \and
    Fukui University of Technology, Fukui, Fukui, Japan\\
    \email{f-oosita@fukui-ut.ac.jp}
}
\maketitle

\begin{abstract}
Self-stabilization
is a versatile methodology 
in the design of fault-tolerant distributed algorithms 
for transient faults.
A self-stabilizing system automatically recovers from 
any kind and any finite number of transient faults.
This property is specifically useful in
modern distributed systems with a large number of components.
In this paper, 
we propose a new communication and execution model named the R(1)W(1) model
in which each process can read and write its own and neighbors'
local variables in a single step. 
We propose
self-stabilizing distributed algorithms in the R(1)W(1) model 
for the problems of
maximal matching,
minimal $k$-dominating set and
maximal $k$-dependent set.
Finally, we propose {an example} transformer, 
based on randomized {distance-two} local mutual exclusion,
to simulate algorithms designed for the R(1)W(1) model
in the synchronous message passing model
with synchronized clocks.
\end{abstract}


\section{Introduction}
\label{SEC:INTRO}

Self-stabilization \cite{Dijkstra1974,dolev00,Altisen2019}
is a versatile methodology for designing 
fault-tolerant distributed algorithms for transient faults.
A transient fault is defined as a corruption of data such as
message corruption, message loss, memory corruption and reboot, 
for example.
A self-stabilizing system automatically recovers from 
any kind and any finite number of transient faults.
It is regarded as a self-organizing system
because a globally synchronized initialization and reset are not necessary and 
the system automatically converges to some legitimate configuration
after the faults. 
This property is specifically useful in
modern distributed systems with a large number of components
such as 
the Internet, 
wireless sensor network, 
ad-hoc network and so on. 
However, 
arbitrary initial configurations and asynchronous executions 
make the design and verification of self-stabilizing distributed algorithms
quite difficult.
In this paper, 
we propose a new communication and execution model named the R(1)W(1) model
which makes the design and verification easier.
Then we propose a {simple} randomized transformer {as an example}
for algorithms designed in the R(1)W(1) model 
under the {unfair} central daemon
to run in the synchronous message passing model with synchronized clocks.

\subsection{Background}

Many self-stabilizing distributed algorithms
adopt a communication model called 
the \emph{state-reading model}
(or the \emph{locally shared memory model}).
This model is introduced in 
the first paper on self-stabilization \cite{Dijkstra1974}, 
and it is widely accepted 
in the research community.
In the state-reading model, 
each process has some \emph{local variables}, and
each process can read local variables of its neighbors
without any delay.
Processes communicate with each other 
by writing values to local variables and 
reading neighbors' local variables.
Furthermore,
many self-stabilizing distributed algorithms
adopt the \emph{composite atomicity model}
(or the \emph{atomic-state model}) 
for modeling executions of processes \cite{Dijkstra1974}. 
In a single \emph{move}, 
each process performs the following three substeps atomically:
(1)~reads its own and neighbors' local variables,
(2)~performs computation based on these values, and
(3)~writes the results on its own local variables.
Asynchronous process execution is modeled by \emph{daemon} \cite{Dijkstra1974}. 
The \emph{central daemon} is a process scheduler that 
selects one process at each step, while
the \emph{distributed daemon} 
selects any non-empty set of processes at each step.
Asynchronous and adversarial process scheduling by daemon 
makes designing self-stabilizing distributed algorithms difficult.
To make algorithm design easier, 
the \emph{distance-two} model \cite{Gairing2004a} and
the \emph{expression} model \cite{Turau2012}
are proposed.
These models enable each process, in a single step,  
to access the local variables of processes 
that are within two hops. 

The models mentioned above seem to be artificial, 
and the self-stabilizing distributed algorithms 
designed under these models do not run 
in real distributed computing environments.
The message passing model is 
closer to actual distributed computing environments, 
however, in general, 
design and verification is difficult in the model
compared to the state-reading model.
Transformation of models is an effective strategy 
for overcoming these difficulties.
An algorithm is designed under a model
such as the distance-two model, and
it is transformed into another model 
such as the message passing model. 

\subsection{Related works}

In the (ordinary) state-reading model,
each process has access to local variables of direct neighbors. 
We call this model the \emph{distance-one} model.
The algorithm design is simplified
by increasing the communication distance of the model, i.e., 
each process is allowed to access to the local variables of processes 
within two or more hops in a single move.  
Existing schemes typically proceed through the following three steps:
(1)~develop a self-stabilizing distributed algorithm assuming 
  the distance-two model \cite{Gairing2004a}
  or 
  the expression state-reading model \cite{Turau2012},  
(2)~transform it to the distance-one state-reading model
    \cite{Gairing2004a,
          Turau2012}, and
(3)~use another transformer 
    \cite{HWT94,
          mizuno96b,
          Herman03}
   to run in the message passing model.

Gairing et al.\ \cite{Gairing2004a}
propose the \emph{distance-two} model for communication.
Each process has access to local variables of processes within two hops 
in a single move.
They also present two transformers that transform 
a self-stabilizing distributed algorithm
in the distance-two model under the central daemon
to the distance-one model under the central and distributed daemons.
The overhead factor of the transformer 
to central (resp., distributed) daemon 
is $m$ (resp., $O(n^{2} m)$),
i.e., 
the time complexity of the transformed algorithm is 
$O(m T)$ (resp. $O(n^{2} m T)$), 
where 
$m$ is the number of edges in the network,  
$n$ is the number of processes, and
$T$ is the time complexity of $A$.

Goddard et al.\ \cite{Goddard2008b}
propose the \emph{distance-$k$ model} for communication
such that 
each process has access to local variables of processes in $k$ hops away,
where $k$ is arbitrary constant.
They also present a transformer which transforms
a self-stabilizing distributed algorithm $A$ 
in the distance-$k$ model under the central daemon
to an algorithm 
in the distance-one model under the central daemon.
The overhead factor of the transformer is 
$O(n^{\log k})$. 

Turau 
\cite{Turau2012}
proposes the \emph{expression model} for communication,
which is a generalization of the distance-two model.
Each process $P_{i}$ has some \emph{expressions}
whose values are determined by local variables of $P_{i}$ and its neighbors, 
and a process has an access to the values of expressions at neighbors.
An expression is considered as an aggregation of local variables of neighbors.
By reading the value of an expression of neighbors,
each process has an access to local variables of processes in two hops.
He proposes two transformers that transform 
a self-stabilizing distributed algorithm 
in the expression model under the central daemon
to the distance-one model in the central and distributed daemons.
The overhead factors of the two transformers are both $O(m)$.

To execute a self-stabilizing distributed algorithm
assuming the distance-one state-reading model 
in a message passing distributed system, 
several methods are proposed
\cite{HWT94,
      mizuno96b,
      Herman03}.
A basic idea which is common to these works is that 
each process has a \emph{cache} of local variables of neighbors, and 
each process reads the cache 
instead of reading local variables located on neighbors.

Another related work for communication model transformation
is the work by Cohen et al.\ \cite{Cohen2023}.
They propose transformers 
from the (distance-one) state-reading model
to the link-register model with read/write atomicity.
A link register is an abstraction of a unidirectional communication channel.
A sender processes writes a value to a link-register and 
a receiver process reads the register. 
Their transformers are based on local mutual exclusion. 


\subsection{Contribution of this paper}


In this paper, 
we propose a new computation model named the \emph{R(1)W(1) model} 
in which 
each process can read and write local variables of direct neighbors
in a single move.
Self-stabilizing algorithms under this model
assume the central daemon only for process scheduler
to avoid simultaneous writes to a local variable
by more than two or more processes. 

To demonstrate the R(1)W(1) model, 
we propose self-stabilizing distributed algorithms 
for the problems of
maximal matching and
minimal $k$-dominating set 
under the unfair central daemon.
The benefit of the proposed model is that 
it makes coordinated actions by neighboring processes simple 
by allowing processes to write neighbors' local variables.

We also propose {an example} transformer 
for silent self-stabilizing distributed algorithms 
in the R(1)W(1) model assuming the unfair central daemon
to the synchronous message passing model with synchronized clocks.
Here, we say that 
an algorithm is \emph{silent} 
if no process never takes any action when the system is stabilized, and
a daemon is \emph{unfair} 
if it takes an arbitrary (adversarial) process scheduling.
The existing transformers for 
the distance-two, distance-$k$ and expression models
generate an algorithm in the distance-one model, and
it needs another conversions to run in the message passing model.
On the other hand,
our transformer immediately generates an algorithm 
in the message passing model.
For simulating the central daemon in the message passing model,
we take an approach by local mutual exclusion based on randomized voting.
Specifically, 
our transformer is based on the \emph{distance-two local mutual exclusion}
to avoid simultaneous moves processes within distance two.
This guarantees that two or more processes never writes 
the same local variable of a process at the same time, and
the R(1)W(1) model is simulated.
We show that at least one process is allowed to take an action
with at least some constant probability.
As we show in Theorem~\ref{LEM:TR:PERFOMANCE}, 
the expected overhead factor of our transformation is 
$O(1)$ in time complexity
and
$O(n)$ in message complexity, 
where $n$ is the number of processes.
On the other hand,
the overhead factor of the transformer by Turau \cite{Turau2012}
is $O(m)$ in time complexity, 
where $m$ is the number of edges, and, unfortunately,
a transformed algorithm needs another model transformer
to run in the message passing model.

\subsection{Organization of this paper}

The rest of this paper is organized as follows.
In Section~\ref{SEC:PRELIMINARY}, 
we introduce the definitions and notation, 
specifically, we propose the R(1)W(1) model.
In Sections~\ref{SEC:SSMAXMAT},
\ref{SEC:SSKDOMSET} and
\ref{SEC:SSKDEPSET},
we propose self-stabilizing distributed algorithms 
in the R(1)W(1) model for problems of
maximal matching, 
minimal $k$-dominating set and 
maximal $k$-dependent set.
In Section~\ref{SEC:TRANSFORMER},
we propose a transformer for algorithm in the R(1)W(1) model
to the synchronous message passing model.
In Section~\ref{SEC:CONCLUSION}
we give concluding remarks.
%

\section{Preliminary}
\label{SEC:PRELIMINARY}

First, 
we define some notations used in this paper.
A distributed system is denoted by a graph $G=(V,E)$, where 
$V$ is the set of processes and
$E \subseteq V \times V$ 
is the set of bidirectional communication links between processes.
The number of processes is denoted by $n ~(= |V|)$.
Processes are denoted by $P_{0}, P_{1}, ..., P_{n-1}$.
The set of neighbor processes of $P_{i}$ is denoted by
$N_{i}$ ~($= \{ P_{j} \in V \mid (P_{i},P_{j}) \in E\}$).
The set of processes in two hops from $P_{i}$ is denoted by
$N^{(2)}_{i}$
($= \{ P_{j} \in V \mid 
       \textnormal{ the distance between } P_{i} \textnormal{ and } P_{j} 
       \textnormal{ is } 2 \}$).
The set of processes within two hops of $P_{i}$ is denoted by
$N^{(1,2)}_{i} ~(= N_{i} \cup N^{(2)}_{i})$.
Each process $P_{i}$ is given, as initial knowledge, 
the values of $N_{i}$, $N^{(2)}_{i}$ and $N^{(1,2)}_{i}$ as constants.

\subsection{The R(1)W(1) model}

In this paper, 
we propose a new computational model, called R(1)W(1), 
which is an extension of 
the ordinary state-reading model. 
In the ordinary state-reading model, 
a single move of each process $P_{i}$ consists of
(1)~reading local variables of $P_{i}$ and processes in $N_{i}$, 
(2)~computing locally, and
(3)~writing to local variables of $P_{i}$.
In the R(1)W(1) model, 
a single move of each process $P_{i}$ consists of
(1)~reading local variables of $P_{i}$ and processes in $N_{i}$, 
(2)~computing locally, and
(3)~writing to local variables of $P_{i}$ and processes in $N_{i}$.
So, a process can update local variables of neighbor processes in a single move.
In this model, 
we assume the central daemon for process scheduler
to avoid simultaneous writes to a local variable
by more than two or more neighbor processes. 
So, we do not assume the distributed daemon.

This model is further generalized to the R($d_{r}$)W($d_{w}$) model 
in which each process 
can read (resp., write) local variables of processes 
within $d_{r}$ (resp., $d_{w}$) hops.
According to our notation, 
the ordinary state-reading model is denoted by R(1)W(0), and
the distance-two model is denoted by R(2)W(0). 

\subsection{Self-stabilization}


Let $q_{i}$ be the local state of process $P_{i} \in V$.
A \emph{configuration} of a distributed system is 
a tuple $(q_{0}, ..., q_{i}, ..., q_{n-1})$
of local states of $P_{0}, ..., P_{i}, ..., P_{n-1}$.
By $\Gamma$, 
we denote the set of all configurations.

We adopt the set of \emph{guarded commands} 
(or, set of \emph{rules})
to describe self-stabilizing distributed algorithms in the R(1)W(1) model
as shown in 
Algorithms~\ref{ALG:MMat11}, 
for example. 
A \emph{guard} 
is a predicate (boolean function)
on local states of processes.
A \emph{command}
is a series of statements to update local variables of process(es). 
We say that a process is \emph{enabled} 
iff it has a guard which evaluates to true.
Otherwise, we say that a process is \emph{disabled}.

We assume that processes are serially scheduled, 
meaning that 
exactly one enabled process is selected and executes a guarded command.
Such a scheduler is called the \emph{central daemon}. 
We assume that the central daemon is \emph{unfair} 
in the sense that the process scheduling may be adversarial,
i.e., it may not select a specific process 
unless the process is the only enabled process.
An enabled process selected by the daemon executes 
a command corresponding to a guard that evaluates to true.
Let $\gamma$ be any configuration, and 
$\gamma'$ be the configuration which follows $\gamma$ in an execution.
Then, this relation is denoted by $\gamma \rightarrow \gamma'$. 
Execution of an algorithm is maximal, 
meaning that 
the execution continues as long as there exists an enabled process.

The correct system states of a distributed system are 
specified by a set of \emph{legitimate} configurations,
denoted by $\Lambda$ ($\subseteq \Gamma$). 

A distributed system is \emph{self-stabilizing} with respect to $\Lambda$
iff the following two conditions are satisfied.
\begin{enumerate}
\item \emph{Closure}:
      For any legitimate configuration $\gamma \in \Lambda$, 
      if there exists an enabled process in $\gamma$, then
      any configuration $\gamma'$ that follows $\gamma$ is also legitimate.
\item \emph{Convergence}:
      For any illegitimate configuration 
      $\gamma \in \Gamma \backslash \Lambda$, then
      configuration of the system becomes legitimate eventually.
\end{enumerate}

\section{Maximal matching in the R(1)W(1) model}
\label{SEC:SSMAXMAT}

In this section,
we propose a self-stabilizing distributed algorithm \textsf{MMat11}
for the maximal matching problem
assuming the R(1)W(1) model under the unfair central daemon.
A matching $F$ of a graph $G=(V,E)$
is a subset of edges $E$ such that,
for each edge $(P_{i},P_{j}) \in F$, 
$(P_{k}, P_{j}) \not\in F$ holds 
for each $P_{k} \in V \backslash\{P_{i}\}$.
A matching $F$ is maximal iff
$F \cup \{(P_{i},P_{j})\}$ is not a matching
for each edge $(P_{i},P_{j}) \in E \backslash F$. 

Self-stabilizing distributed algorithms for 
the maximal matching problem 
have been proposed.
To represent the time complexities of algorithms, 
we adopt the total number of moves (or steps)
which counts the total number of executions of guarded commands to converge.
Hedetniemi et al.\ 
proposed an algorithm 
with time complexity $O(m)$
under the unfair central daemon in \cite{Hedetniemi2001}.
Manne et al.\ 
proposed an algorithm 
with time complexity $O(m)$
under the unfair distributed daemon in \cite{Manne2009}.
On the other hand, 
the time complexity
of our algorithm \textsf{MMat11} is $O(n)$. 

\subsection{The proposed algorithm \textsf{MMat11}}

The proposed algorithm \textsf{MMat11} is presented
in Algorithm~\ref{ALG:MMat11}.
Each process $P_{i}$ maintains a single local variable $q_{i}$.
We say that $P_{j} \in N_{i}$ is a \emph{matching neighbor} of $P_{i}$
iff $P_{j} = q_{i}$ and $P_{i} = q_{j}$ hold.
If $P_{j}$ is a matching neighbor of $P_{i}$, 
we say that $P_{i}$ and $P_{j}$ are \emph{matching pair}.
We say that $P_{i}$ is \emph{free} 
iff $q_{i} = \bot$ holds.
We say that $P_{i}$ \emph{points} to $P_{j} \in N_{i}$
iff $q_{i} = P_{j}$.

There are five rules in \textsf{MMat11}.
\begin{itemize}
\item Rule~1: 
      If $P_{i}$ is free and it is pointed by $P_{j}$, 
      then
      $P_{i}$ accepts the proposal of $P_{j}$, and
      $P_{i}$ becomes a matching neighbor of $P_{j}$.
\item Rule~2:
      If $P_{i}$ is free and 
      there exists a free neighbor $P_{j}$,
      then 
      $P_{i}$ forces $P_{j}$ to become a matching neighbor of $P_{i}$.
\item Rule~3:
      If $P_{j}$ to which $P_{i}$ points is free, 
      then 
      $P_{i}$ forces $P_{j}$ to become a matching neighbor of $P_{i}$.
\item Rule~4:
      If $P_{j}$ to which $P_{i}$ points does not point to $P_{i}$
      but there exists 
      a neighbor $P_{k} \in N_{i}$ which is free or $q_{k}=P_{i}$ holds,
      then 
      $P_{i}$ becomes a matching neighbor of $P_{k}$.
      In the former case, $P_{i}$ forces $P_{k}$ to point to $P_{i}$.
\item Rule~5:
      If $P_{j}$ to which $P_{i}$ points does not point to $P_{i}$
      and 
      each neighbor $P_{k} \in N_{i}$
      is not free and does not point to $P_{i}$,
      then 
      $P_{i}$ gives up finding a matching neighbor.
\end{itemize}

By Rules~1, 2, 3 or 4, 
$P_{i}$ makes a matching pair with a neighbor, and
the matching pair is maintained forever.

\begin{algorithm}
\caption{Self-stabilizing distributed maximal matching algorithm \textsf{MMat11}}
\label{ALG:MMat11}
\KwLVar\\
\Indp
$q_{i} \in N_{i} \cup \{\bot\}$ 
  \tcp*{the matching neighbor of $P_{i}$}
\Indm
\BlankLine
\Fn(\tcp*[h]{Free $P_{i}$ accepts $P_{j}$, and make a matching})%
   {\normalshape Rule~1}{
  \If{%
     $q_{i} = \bot 
        \land 
      \exists P_{j} \in N_{i} : q_{j} = P_{i}$
  }{%
    $q_{i} := P_{j}$
  }
}
\Fn(\tcp*[h]{Force $P_{i}$ and $P_{j}$ to make a matching})%
   {\normalshape Rule~2}{
  \If{%
    $q_{i} = \bot 
       \land 
     \exists P_{j} \in N_{i} : q_{j} = \bot$
  }{%
    $q_{i} := P_{j}; q_{j} := P_{i}$
  }
}
\Fn(\tcp*[h]{Force $P_{j}$ to make a matching with $P_{i}$})%
   {\normalshape Rule~3}{
  \If{%
    $q_{i} \in N_{i} 
     \land 
     q_{j} = \bot, \textnormal{ where } P_{j} = q_{i}$ 
  }{%
    $q_{j} := P_{i}$
  }
}
\Fn(\tcp*[h]{$P_{i}$ switches to $P_{k}$ to make a matching})%
   {\normalshape Rule~4}{
  \If{%
    $q_{i} \in N_{i} 
     \land 
     q_{j} \not\in \{P_{i},\bot\}, 
     \textnormal{ where } P_{j} = q_{i},
       \break
       {} \land 
     \exists P_{k} \in N_{i} : q_{k} \in \{P_{i},\bot\}$
  }{%
    $q_{i} := P_{k}$; $q_{k} := P_{i}$
  }
}
\Fn(\tcp*[h]{$P_{i}$ gives up})%
   {\normalshape Rule~5}{
  \If{%
    $q_{i} \in N_{i} 
       \land 
     q_{j} \not\in \{P_{i},\bot\}, 
     \textnormal{ where } P_{j} = q_{i},
     \break 
       {} \land
     \forall P_{k} \in N_{i} : q_{k} \not\in \{P_{i},\bot\}$
  }{%
    $q_{i} := \bot$
  }
}
\end{algorithm}

\subsection{The proof of correctness of \textsf{MMat11}}

By $\Gamma_{\textrm{MM}}$, 
we denote the set of all configurations of \textsf{MMat11}.
A configuration $\gamma$ of \textsf{MMat11} is legitimate 
iff the following two conditions are satisfied:
\begin{itemize}
\item Matching: 
      $\forall P_{i} \in V : 
         q_{i} \in N_{i}
           \Rightarrow
         q_{j} = P_{i}
      $, where $P_{j} = q_{i}$.
\item Maximality: 
      $\forall P_{i} \in V : 
         q_{i} = \bot 
           \Rightarrow
         (\forall P_{j} \in N_{i} : q_{j} \not\in \{P_{i},\bot\}) $.
\end{itemize}
By $\Lambda_{\textrm{MM}}$, 
we denote the set of legitimate configurations of \textsf{MMat11}.
Let 
  $F_{\textrm{MM}}(\gamma) 
   = \{ (P_{i},P_{j}) \in E \mid q_{i}=P_{j} \land q_{j}=P_{i}\}$
be the set of matching pairs.

\begin{lemma}
\label{LEM:MM:LEG_CONFIG-MM}
For each $\gamma \in \Lambda_{\textrm{MM}}$, 
$F_{\textrm{MM}}(\gamma)$ is a maximal matching of $G$.
\end{lemma}
\begin{proof}
Let $\gamma$ be any configuration in $\Lambda_{\textrm{MM}}$.
First, we show that 
$F_{\textrm{MM}}(\gamma)$ is a matching of $G$.
For each $P_{i} \in V$, 
by the definition of legitimate configuration,
if $q_{i} \in N_{i}$ then $q_{j} = P_{i}$ holds, 
where $P_{j} = q_{i}$, 
i.e., 
$q_{i} \in N_{i}$ implies $(P_{i},P_{j}) \in F_{\textrm{MM}}(\gamma)$.
Because there exists no two distinct processes $P_{j}$ and $P_{k}$ 
such that $(P_{i},P_{j}), (P_{i},P_{k}) \in F_{\textrm{MM}}(\gamma)$, 
$F_{\textrm{MM}}(\gamma)$ is a matching.
Next, we show that 
a matching $F_{\textrm{MM}}(\gamma)$ of $G$ is maximal.
For each $P_{i} \in V$, 
by the definition of legitimate configurations,
if $q_{i} = \bot$ then $\forall P_{j} \in N_{i} : q_{j} \ne P_{i}$ holds, 
i.e., 
there exists no two processes $P_{i}$ and $P_{j}$ 
such that $F_{\textrm{MM}}(\gamma) \cup \{ (P_{i},P_{j}) \}$ 
is a matching of $G$.
Hence $F_{\textrm{MM}}(\gamma)$ is maximal.
\qed
\end{proof}

\begin{lemma}
\label{LEM:MM:CLOSURE}
(Closure)
Every process is disabled in $\gamma$
iff $\gamma \in \Lambda_{\textrm{MM}}$. 
\end{lemma}
\begin{proof}
($\Rightarrow$)
Let $\gamma \in \Gamma_{\textrm{MM}}$ be any configuration
such that every process is disabled in $\gamma$, and
$P_{i} \in V$ be any process. 
In the case $q_{i} = \bot$ holds,
by \textsf{MMat11}, 
$\forall P_{j} \in N_{i} : q_{j} \ne P_{i} \land q_{j} \ne \bot$
holds, 
which is equivalent to the maximality condition of legitimate configurations.
In the case $q_{i} = P_{j} \in N_{i}$ holds,
by \textsf{MMat11}, 
$q_{j} \ne \bot \land q_{j} \in \{P_{i},\bot\}$
holds,
which is equivalent to the matching condition of legitimate configurations.
Hence $\gamma \in \Lambda_{\textrm{MM}}$ holds. 

($\Leftarrow$)
Let $\gamma \in \Lambda_{\textrm{MM}}$ be any legitimate configuration, and 
$P_{i} \in V$ be any process.
In the case $q_{i} = \bot$ in $\gamma$ holds, 
by the maximality condition of legitimate configurations, 
$\forall P_{j} \in N_{i} : q_{j} \not\in \{P_{i},\bot\}$ holds, and
$P_{i}$ is not enabled by Rules~1 and 2.
Obviously, $P_{i}$ is not enabled by Rules~3, 4 and 5 in this case.
In the case $q_{i} = P_{j} \in N_{i}$ in $\gamma$ holds,  
by the matching condition of legitimate configurations, 
$q_{j} = P_{i}$ holds, and
$P_{i}$ is not enabled by Rules~3, 4 and 5.
Obviously, $P_{i}$ is not enabled by Rules~1 and 2 in this case.
\qed
\end{proof}

Let 
$A(\gamma) = |F_{\textrm{MM}}(\gamma)|$ and
$B(\gamma) = |\{ P_{i} \in V \mid 
                 q_{i} \in N_{i} \land 
                 q_{j} \not\in\{P_{i},\bot\}, 
                 \textnormal{where } P_{j} = q_{i} \}|$.
Intuitively speaking, 
$A$ represents the number of matching pairs, and
$B$ represents the number of processes $P_{i}$ such that 
the value of $q_{i}$ is incorrect.
For any configuration $\gamma \in \Gamma_{\textrm{MM}}$, 
$0 \leq A(\gamma) \leq \lfloor n/2 \rfloor$ and
$0 \leq B(\gamma) \leq n$ hold.

\begin{lemma}
\label{LEM:MM:A,B:MONOTONIC}
For any $\gamma, \gamma' \in \Gamma_{\textrm{MM}}$ 
such that $\gamma \rightarrow \gamma'$, 
i.e., $\gamma$ is not legitimate,
$A(\gamma) \leq A(\gamma')$ and 
$B(\gamma) \geq B(\gamma')$ hold.
Furthermore, 
$A(\gamma) < A(\gamma')$ or
$B(\gamma) > B(\gamma')$ holds.
\end{lemma}
\begin{proof}
A move by Rules~1,2,3 or 4 
increases the value of $A$ by one, however,
a move by Rule~5 does not.
A move by Rules~4 or 5
decreases the value of $B$ by one, however,
a move by Rules~1, 2 or 3 does not.
For any move, 
$A(\gamma) = A(\gamma')$ and
$B(\gamma) = B(\gamma')$ do not occur at the same time.
\qed
\end{proof}

\begin{lemma}
\label{LEM:MM:CONVERGENCE}
(Convergence)
Starting from arbitrary configuration in $\Gamma_{\textrm{MM}}$, 
any execution of \textsf{MMat11} reaches 
a legitimate configuration $\gamma \in \Lambda_{\textrm{MM}}$.
\end{lemma}
\begin{proof}
By Lemma~\ref{LEM:MM:A,B:MONOTONIC}, 
any move changes the values of at least one of $A$ or $B$.
Because $A$ and $B$ are bounded,
there exists no infinite execution.
Hence any execution is finite and 
terminates in which no process is enabled.
By Lemmas~\ref{LEM:MM:LEG_CONFIG-MM} and \ref{LEM:MM:CLOSURE},
such a configuration is legitimate. 
\qed
\end{proof}

\begin{theorem}
\label{THR:MM:SS}
\textsf{MMat11} is self-stabilizing 
with respect to $\Lambda_{\textrm{MM}}$
under the unfair central daemon in the R(1)W(1) model, and
its time complexity is $O(n)$.
\end{theorem}
\begin{proof}
By Lemmas~\ref{LEM:MM:CLOSURE} and \ref{LEM:MM:CONVERGENCE},
\textsf{MMat11} is self-stabilizing. 
Because 
$0 \leq A(\gamma) \leq \lfloor n/2 \rfloor, 
 0 \leq B(\gamma) \leq n$ 
hold for any initial configuration $\gamma$, and
any move changes the value of at least one of $A$ or $B$
by Lemma~\ref{LEM:MM:A,B:MONOTONIC}, 
the maximum number of moves is bounded by 
$\lfloor n/2 \rfloor + n  =  O(n)$.
\qed
\end{proof}

\section{Minimal $k$-dominating set in the R(1)W(1) model}
\label{SEC:SSKDOMSET}

In this section,
we propose a self-stabilizing distributed algorithm \textsf{MkDom11}
for the minimal $k$-dominating set problem
assuming the R(1)W(1) model under the central daemon.
For each integer $k \geq 1$, 
a $k$-dominating set $S$ of a graph $G=(V,E)$
is a subset of vertices $S \subseteq V$ such that, 
for each vertex $P_{i} \in V \backslash S$, 
$|N_{i} \cap S| \geq k$ holds.
A $k$-dominating set $S$ is minimal iff
any proper subset of $S$ is not a $k$-dominating set.
The definition is a generalization of 
the minimal dominating set (MDS), 
i.e., 
the definitions of
the minimal $1$-dominating set and the minimal dominating set 
are equivalent.

An $S \subseteq V$ is a minimal $k$-dominating set 
iff 
the following local conditions hold for each $P_{i} \in V$, 
and we design a distributed algorithm based on these local conditions. 
\begin{itemize}
\item Local $k$-Domination:
      $P_{i} \in V \backslash S \Rightarrow
           |\{ P_{j} \in N_{i} \cap S \}| \geq k$
\item Local Minimality: 
      $P_{i} \in S \Rightarrow
           \exists P_{j} \in N_{i} \cap (V \backslash S) : 
                   |\{ P_{k} \in N_{j} \cap S \}| \leq k$
\end{itemize}

Many self-stabilizing distributed algorithms for 
the dominating set problem are proposed. 
Below, algorithms not explicitly mentioned
assume the ordinary state-reading model.
Hedetniemi et al.\ \cite{Hedetniemi2003}
proposed an algorithm
for the special case $k=1$, which is equivalent to MDS. 
Kamei and Kakugawa \cite{Kamei2003}, 
proposed an algorithm in tree networks
in the general case of $k>1$.  
In the general case of $k>1$ and in general networks, 
Wang et al.\ \cite{Wang2012}
proposed an algorithm under the central daemon, and 
its time complexity is $O(n^{2})$.

Turau \cite{Turau2012} 
proposed an algorithm
in the general case of $k \geq 1$
in the expression model under the central daemon, and 
its time complexity is $O(n)$.
In this section, for the general case of $k \geq 1$,
we propose an algorithm in the R(1)W(1) model under the central daemon
whose time complexity is $O(n)$.

\subsection{The proposed algorithm \textsf{MkDom11}}

The proposed algorithm \textsf{MkDom11} is presented
in Algorithm~\ref{ALG:MkDom11}.
Each process $P_{i}$ maintains two local variables
$x_{i}$ and $c_{i}$.
$P_{i}$ is in a $k$-dominating set iff $x_{i} = 1$, and
$c_{i}$ counts the number of neighbors $P_{j}$ such that $x_{j} = 1$.
We define a macro $\textit{Count}_{i}()$ which represents 
the number of neighbors $P_{j}$ such that $x_{j}=1$.
We say that $c_{i}$ is \emph{correct}
iff $c_{i} = \textit{Count}_{i}()$ holds.

The value of $c_{i}$ is maintained to be equal to $\textit{Count}_{i}()$
so that neighbors of $P_{i}$ can read the value of $\textit{Count}_{i}()$.
In other words, 
$c_{i}$ gives an aggregated information of distance-two processes
to neighbors of $P_{i}$.
To maintain $c_{i}$ to be correct in the R(1)W(1) model,
each neighbor $P_{j}$ increments (resp., decrements) $c_{i}$ by one
when $P_{j}$ changes the value of $x_{j}$ 
from 0 to 1 (resp., 1 to 0). 
Then, 
once $c_{i}$ becomes correct, 
neighbors of $P_{i}$ maintains correctness of $c_{i}$ thereafter.

There are three rules in \textsf{MkDom11}.
\begin{itemize}
\item Rule~1:
      If $c_{i}$ is incorrect, $P_{i}$ fixes it.
\item Rule~2:
      This is a rule for local $k$-domination condition.
      If $x_{i}=0$ and
      the number of neighbors $P_{j}$ such that $x_{j} = 1$ is less than $k$, 
      $P_{i}$ changes $x_{i}$ from 0 to 1
      in order to satisfy the local $k$-domination condition.
      In addition, 
      $P_{i}$ increments $c_{j}$ by one for each neighbor $P_{j}$, 
      however,
      $P_{i}$ does not increment $c_{j}$ if $c_{j} \geq |N_{j}|$ holds
      because $c_{j}$ is obviously incorrect. 
      Here, we implicitly assume that 
      $P_{i}$ has access to the value of $|N_{j}|$, 
      which can be implemented by a local variable at $P_{j}$ 
      to hold the value. 
\item Rule~3:
      This is a rule for the local minimality condition.
      $P_{i}$ changes $x_{i}$ from 1 to 0
      if such a change does not violate 
      the local $k$-domination condition. 
      If $P_{i}$ changes $x_{i}$, 
      it decrements $c_{j}$ by one for each neighbor $P_{j}$, 
      however,
      $P_{i}$ does not decrement if $c_{j} = 0$ holds
      because $c_{j}$ is obviously incorrect. 
\end{itemize}

\begin{algorithm}
\caption{Self-stabilizing distributed minimal $k$-dominating set algorithm \textsf{MkDom11}}
\label{ALG:MkDom11}
\KwLVar\\
\Indp
  $x_{i} \in \{1, 0\}$ 
    \tcp*{whether a member of the set or not}
  $c_{i} \in \{0,1,..., |N_{i}|\}$ 
    \tcp*{\#neighbors s.t. $x_{j}=1$}
\Indm
\KwMac\\
\Indp
  $\textit{Count}_{i}() \equiv |\{ P_{j} \in N_{i} \mid x_{j} = 1 \}|$ \;
\Indm
\BlankLine
\Fn(\tcp*[h]{Fix the counter})%
   {\normalshape Rule~1}{
  \If{%
     $c_{i} \ne \textit{Count}_{i}()$
  }{%
    $c_{i} := \textit{Count}_{i}()$
  }
}
\Fn(\tcp*[h]{$k$-Domination})%
   {\normalshape Rule~2}{
  \If{%
    $x_{i} = 0 
       \land 
     c_{i} = \textit{Count}_{i}()
       \land 
     c_{i} < k$
  }{%
    $x_{i} := 1$ \;
    for each $P_{j} \in N_{i}$ s.t. $c_{j} < |N_{j}|$ : \;
    \quad  $c_{j} := c_{j} + 1$
  }
}
\Fn(\tcp*[h]{Minimality})%
   {\normalshape Rule~3}{
  \If{%
    $x_{i} = 1 
       \land 
     c_{i} = \textit{Count}_{i}()
       \land 
     c_{i} \geq k
       \land 
     (\forall P_{j} \in N_{i} : x_{j} = 1 \lor c_{j} > k)$
  }{%
    $x_{i} := 0$ \;
    for each $P_{j} \in N_{i}$ s.t. $c_{j} > 0$ : \;
    \quad  $c_{j} := c_{j} - 1$
  }
}
\end{algorithm}

\subsection{The proof of correctness of \textsf{MkDom11}}

By $\Gamma_{\textnormal{MkDom}}$, 
we denote the set of all configurations of \textsf{MkDom11}.
A configuration $\gamma$ of \textsf{MkDom11} is legitimate
iff the following three conditions are satisfied
for each $P_{i} \in V$.
\begin{itemize}
\item Correctness of the count: 
      $c_{i} = \textit{Count}_{i}()$
\item Local $k$-Domination: 
      $x_{i}=0 \Rightarrow
           c_{i} \geq k $
\item Local Minimality:
      $x_{i}=1 \Rightarrow
          c_{i} < k \lor
          \exists P_{j} \in N_{i} : x_{j} = 0 \land c_{j} \leq k$
\end{itemize}

By $\Lambda_{\textnormal{MkDom}}$, 
we denote the set of legitimate configurations of \textsf{MkDom11}.

\begin{lemma}
\label{LEM:MKDOM:CLOSURE}
A configuration $\gamma$ is legitimate 
iff 
no process is enabled. 
\end{lemma}
\begin{proof}
($\Rightarrow$)
Because $c_{i}$ is correct, 
$P_{i}$ is not enabled by Rule~1. 
Because $c_{i}$ is correct and the $k$-domination condition
$x_{i}=0 \Rightarrow c_{i} \geq k$
holds, 
$P_{i}$ is not enabled by Rule~2. 
Because $c_{i}$ is correct and the minimality condition
$x_{i}=1 \Rightarrow 
  c_{i} < k \lor 
  \exists P_{j} \in N_{i} : x_{j} = 0 \land c_{j} \leq k$
holds, 
$P_{i}$ is not enabled by Rule~3.

($\Leftarrow$)
By Rule~1, 
$c_{i} = \textit{Count}_{i}()$ holds. 
By Rule~2, 
if $x_{i} = 0$ then 
$c_{i} \geq k$ holds.
Hence the $k$-domination condition holds.
By Rule~3, 
if $x_{i} = 1$ then 
$c_{i} < k$ or $\exists P_{j} \in N_{i} : x_{j} = 0 \land c_{j} \geq k$ hold.
Hence the minimality condition holds.
\qed
\end{proof}

\begin{lemma}
\label{LEM:MKDOM:C-EQ-COUNT}
For each process $P_{i} \in V$, 
if the condition 
$c_{i} = \textit{Count}_{i}()$ 
holds, 
it remains so thereafter.
\end{lemma}
\begin{proof}
For each neighbor $P_{j} \in N_{i}$, 
when $P_{j}$ changes $x_{j}$ from 0 to 1 (resp. 1 to 0), 
$P_{j}$ increments (resp. decrements) $c_{i}$ by one.
Hence, if the condition $c_{i} = \textit{Count}_{i}()$ holds, 
it remains so thereafter.
\qed
\end{proof}

\begin{lemma}
\label{LEM:MKDOM:R1:ONCE}
For each process $P_{i} \in V$, 
the number of moves by Rule~1
is at most once, and
if $P_{i}$ moves by Rule~1, 
it is the first move of $P_{i}$.
\end{lemma}
\begin{proof}
In case Rule~1 
is the rule of $P_{i}$'s first move,
the condition $c_{i} = \textit{Count}_{i}()$ becomes true
and it remains so thereafter
by Lemma~\ref{LEM:MKDOM:C-EQ-COUNT}.
Hence $P_{i}$ never moves by Rule~1 again.

In case Rule~2 or 3 is the rule of $P_{i}$'s first move,
the condition $c_{i} = \textit{Count}_{i}()$ holds
before $P_{i}$ moves by Rule~2 or 3.
By Lemma~\ref{LEM:MKDOM:C-EQ-COUNT},
the condition holds thereafter, and 
hence $P_{i}$ never moves by Rule~1.
\qed
\end{proof}

\begin{lemma}
\label{LEM:MKDOM:R3:ONCE}
For each process $P_{i} \in V$, 
the number of moves by Rule~3
is at most once.
\end{lemma}
\begin{proof}
Suppose that $P_{i}$ moves by Rule~3.
After the move, we have
$x_{i} = 0$, 
$c_{i} = \textit{Count}_{i}()$ and
$c_{i} \geq k$.
Before $P_{i}$ moves by Rule~3 for the second time, 
$P_{i}$ must move by Rule~2.
Hence $c_{i} < k$, which is a part of the guard of Rule~2, 
must be true at $P_{i}$.
Because $c_{i} \geq k$ holds 
before $P_{i}$ moves by Rule~3 for the first time,
one or more neighbors $P_{j} \in N_{i}$ must move by Rule~3 
in order to satisfy the condition $c_{i} < k$.

When $c_{i}>k$ holds, 
some neighbor $P_{j}$ may move by Rule~3, and
the value of $c_{i}$ decreases.
However, when $c_{i}=k$ holds, 
the guard of Rule~3 is false at any neighbor $P_{j}$, and
no neighbor moves by Rule~3 any more.
Hence $c_{i}<k$ never becomes true, and
$P_{i}$ does not move by Rule~2, 
which means that
$P_{i}$ does not move by Rule~3 again.
\qed
\end{proof}

\begin{lemma}
\label{LEM:MKDOM:R2:TWICE}
For each process $P_{i} \in V$, 
the number of moves by Rule~2
is at most twice.
\end{lemma}
\begin{proof}
For $P_{i}$ to move by Rule~2 three times, 
$P_{i}$ must move by Rule~3 twice.
But it is impossible 
by Lemma~\ref{LEM:MKDOM:R3:ONCE}.
\qed
\end{proof}

\begin{theorem}
\label{THR:MkDom:SS}
\textsf{MkDom11} is self-stabilizing 
with respect to $\Lambda_{\textrm{MkDom}}$
under the unfair central daemon in the R(1)W(1) model, and
its time complexity is $O(n)$.
\end{theorem}
\begin{proof}
The closure condition holds
by Lemma~\ref{LEM:MKDOM:CLOSURE}. 
The convergence condition holds
because the number of moves is bounded at each process. 
By lemmas~\ref{LEM:MKDOM:R1:ONCE}, 
\ref{LEM:MKDOM:R3:ONCE} and
\ref{LEM:MKDOM:R2:TWICE}, 
each process $P_{i} \in V$ moves
by Rule~1 at most once, 
by Rule~3 at most once, and
by Rule~2 at most twice.
Hence $P_{i}$ moves at most four times, and
the total number of moves is bounded by $4n$.
\qed
\end{proof}

\section{Maximal $k$-dependent set in the R(1)W(1) model}
\label{SEC:SSKDEPSET}

In this section,
we propose a self-stabilizing distributed algorithm \textsf{MkDep11}
for the maximal $k$-dependent set problem
assuming the R(1)W(1) model under the unfair central daemon.
For each integer $k \geq 0$, 
a $k$-dependent set $S$ of a graph $G=(V,E)$
is a subset of vertices $S \subseteq V$ such that,
for each vertex $P_{i} \in S$, 
$|N_{i} \cap S| \leq k$ holds.
A $k$-dependent set $S$ is maximal iff
any superset of $S$ is not a $k$-dependent set.
The definition is a generalization of 
maximal independent set (MIS), 
i.e., 
the definitions of
maximal $0$-dependent set and maximal independent set 
are equivalent.

An $S \subseteq V$ is a maximal $k$-dependent set 
iff 
the following local conditions hold for each $P_{i} \in V$, 
and we design a distributed algorithm based on these local conditions. 
\begin{itemize}
\item Local $k$-Dependency: 
      $P_{i} \in S \Rightarrow
           |\{ P_{j} \in N_{i} \cap S \}| \leq k$.
\item Local Maximality:
      $P_{i} \in V \backslash S \Rightarrow
           \exists P_{j} \in N_{i} \cap S : 
                   |\{ P_{k} \in N_{j} \cap S \}| \geq k$.
\end{itemize}

Several self-stabilizing distributed algorithms for 
the $k$-dependent set problem are proposed. 
For the case of $k=0$, which is equivalent to MIS, 
Shukla et al.\ \cite{Shukla1995},
Ikeda et al.\ \cite{Ikeda2002} and 
Turau \cite{Turau2007}  
proposed algorithms in the ordinary state-reading model.
For general case of $k > 0$,
Turau \cite{Turau2012}
proposed an algorithm in the expression model under the central daemon, and
its time complexity is $O(n)$.
In this section, for the general case of $k \geq 0$, 
we propose an algorithm in the R(1)W(1) model under the central daemon
whose time complexity is $O(n)$.

\subsection{The proposed algorithm \textsf{MkDep11}}

The proposed algorithm \textsf{MkDep11} is presented
in Algorithm~\ref{ALG:MkDep11}.
Each process $P_{i}$ maintains two local variables
$x_{i}$ and $c_{i}$.
$P_{i}$ is in a $k$-dependent set iff $x_{i} = 1$, and
$c_{i}$ counts the number of neighbors $P_{j}$ such that $x_{j} = 1$.
We define a macro $\textit{Count}_{i}()$ which represents 
the number of neighbors $P_{j}$ such that $x_{j}=1$.
We say that $c_{i}$ is \emph{correct}
iff $c_{i} = \textit{Count}_{i}()$ holds.

The value of $c_{i}$ is maintained to be equal to $\textit{Count}_{i}()$, 
however,
it may not in the initial configuration
because of the self-stabilizing problem setting.
In the ordinary state-reading model, 
even if $c_{i}$ is equal to $\textit{Count}_{i}()$, 
it immediately becomes unequal
if a neighbor $P_{j}$ of $P_{i}$ changes the value of $x_{j}$.
To maintain $c_{i}$ to be correct in the R(1)W(1) model,
$P_{i}$ increments (resp., decrements) $c_{j}$ by one
for each neighbor $P_{j}$
when $P_{i}$ changes the value of $x_{j}$ 
from 0 to 1 (resp., 1 to 0). 
Then, 
if $c_{i}$ becomes correct,  
$c_{j}$ is maintained correctly thereafter.

There are three rules in \textsf{MkDep11}.
\begin{itemize}
\item Rule~1:
      If $c_{i}$ is incorrect, $P_{i}$ fixes it.
\item Rule~2:
      This is a rule for local $k$-dependency condition.
      If $x_{i}=1$ and
      the number of neighbors $P_{j}$ such that $x_{j} = 1$ is more than $k$, 
      $P_{i}$ changes $x_{i}$ from 1 to 0
      in order to satisfy the local $k$-dependency condition.
      In addition, 
      $P_{i}$ decrements $c_{j}$ by one for each neighbor $P_{j}$, 
      however,
      $P_{i}$ does not for $P_{j}$ such that $c_{j} = 0$ 
      because $c_{j}$ is obviously incorrect. 
\item Rule~3:
      This is a rule for maximality condition.
      $P_{i}$ changes $x_{i}$ from 0 to 1
      if such a change does not violate 
      the local condition of $k$-dependency. 
      If $P_{i}$ changes $x_{i}$, 
      it increments $c_{j}$ by one for each neighbor $P_{j}$, 
      however,
      $P_{i}$ does not for $P_{j}$ such that $c_{j} \geq |N_{j}|$ 
      because $c_{j}$ is obviously incorrect. 
\end{itemize}

\begin{algorithm}
\caption{Self-stabilizing distributed maximal $k$-dependent set algorithm \textsf{MkDep11}}
\label{ALG:MkDep11}
\KwLVar\\
\Indp
  $x_{i} \in \{1, 0\}$ 
    \tcp*{whether a member of the set or not}
  $c_{i} \in \{0,1,..., |N_{i}|\}$ 
    \tcp*{\#neighbors s.t. $x_{j}=1$}
\Indm
\KwMac\\
\Indp
  $\textit{Count}_{i}() \equiv |\{ P_{j} \in N_{i} \mid x_{j} = 1 \}|$ \;
\Indm
\BlankLine
\Fn(\tcp*[h]{Fix the counter})%
   {\normalshape Rule~1}{
  \If{%
     $c_{i} \ne \textit{Count}_{i}()$
  }{%
    $c_{i} := \textit{Count}_{i}()$
  }
}
\Fn(\tcp*[h]{$k$-Dependency})%
   {\normalshape Rule~2}{
  \If{%
    $x_{i} = 1 
       \land 
     c_{i} = \textit{Count}_{i}()
       \land 
     c_{i} > k$
  }{%
    $x_{i} := 0$ \;
    for each $P_{j} \in N_{i}$ s.t. $c_{j} > 0$: \;
    \quad  $c_{j} := c_{j} - 1$
  }
}
\Fn(\tcp*[h]{Maximality})%
   {\normalshape Rule~3}{
  \If{%
    $x_{i} = 0 
       \land 
     c_{i} = \textit{Count}_{i}()
       \land 
     c_{i} \leq k
       \land 
     (\forall P_{j} \in N_{i} : x_{j} = 0 \lor c_{j} < k)$
  }{%
    $x_{i} := 1$ \;
    for each $P_{j} \in N_{i}$ s.t. $c_{j} < |N_{j}|$: \;
    \quad  $c_{j} := c_{j} + 1$
  }
}
\end{algorithm}

\subsection{The proof of correctness of \textsf{MkDep11}}

By $\Gamma_{\textnormal{MkDep}}$, 
we denote the set of all configurations of \textsf{MkDep11}.
A configuration $\gamma$ of \textsf{MkDep11} is legitimate
iff the following three conditions are satisfied
for each $P_{i} \in V$.
\begin{itemize}
\item Correctness of the count: 
      $c_{i} = \textit{Count}_{i}()$
\item $k$-Dependency: 
      $x_{i}=1 \Rightarrow
           c_{i} \leq k $.
\item Maximality:
      $x_{i}=0 \Rightarrow
          c_{i} > k \lor
          \exists P_{j} \in N_{i} : x_{j} = 1 \land c_{j} \geq k$.
\end{itemize}

By $\Lambda_{\textnormal{MkDep}}$, 
we denote the set of legitimate configurations of \textsf{MkDep11}.

\begin{lemma}
\label{LEM:MKDEP:CLOSURE}
A configuration $\gamma$ is legitimate 
iff 
no process is enabled. 
\end{lemma}

\begin{lemma}
\label{LEM:MKDEP:C-EQ-COUNT}
For each process $P_{i} \in V$, 
if the condition 
$c_{i} = \textit{Count}_{i}()$ 
holds, 
it remains so thereafter.
\end{lemma}

\begin{lemma}
\label{LEM:MKDEP:R1:ONCE}
For each process $P_{i} \in V$, 
the number of moves by Rule~1
is at most once, and
if $P_{i}$ executes Rule~1, 
it is the first move of $P_{i}$.
\end{lemma}

\begin{lemma}
\label{LEM:MKDEP:R3:ONCE}
For each process $P_{i} \in V$, 
the number of moves by Rule~3
is at most once.
\end{lemma}

\begin{lemma}
\label{LEM:MKDEP:R2:TWICE}
For each process $P_{i} \in V$, 
the number of moves by Rule~2
is at most twice.
\end{lemma}

\begin{theorem}
\label{THR:MkDep:SS}
\textsf{MkDep11} is self-stabilizing 
with respect to $\Lambda_{\textrm{MkDep}}$
under the unfair central daemon in the R(1)W(1) model, and
its time complexity is $O(n)$.
\end{theorem}

\section{The transformer to the message passing model}
\label{SEC:TRANSFORMER}

In this section,
we propose {an example of} a transformer \textsf{TrR1W1} for 
a self-stabilizing algorithm in the R(1)W(1) model
to execute in the synchronous message passing model.
The transformer adopts randomized voting mechanism 
to simulate the state-reading model and the central daemon.
The proposed transformer is presented in Algorithm~\ref{TR:ALG:TRR1W1}.
We use the following terms:
a \emph{target algorithm} (e.g., \textsf{MMat11}) 
is an algorithm in the R(1)W(1) model to be simulated, and 
a \emph{transformed algorithm} 
is an algorithm in the synchronous message passing model 
transformed by our transformer \textsf{TrR1W1}.
We assume a network $G=(V,E)$ of processes 
$V = \{ {\cal P}_{0}, {\cal P}_{1}, ... {\cal P}_{n-1} \}$
in the synchronous message passing model.
Each process ${\cal P}_{i}$ simulates 
$P_{i}$ of a target algorithm.

Let us explain the computational model.
We assume a synchronous message passing distributed model
with reliable communication.
Execution of processes are synchronized in \emph{round}. 
In each round, 
each process synchronously
sends a message by {\XBCAST} primitive, 
receives all messages from neighbors, and
updates its local variables by local computation.
The {\XBCAST} primitive broadcasts a message to direct neighbors, and
it is reliable, 
i.e.,  
each message sent by {\XBCAST} is not lost and received by direct neighbors.
In the self-stabilizing setting,
the assumption on the reliability of communication
may seem to be inadequate.
However, 
after the transformed target algorithm converges, 
any message loss does not break the legitimate configuration.
So, it is enough to assume that 
the communication is reliable during convergence.
The proposed transformer is described as a series of \emph{phases},  
each of which corresponds to a round of the synchronous execution model. 
We assume a synchronized clock
is available for each process, and
all processes execute the same phase at the same time.
(The transformer presented later consists of series of five phases, and
we call these fives phases \emph{cycle}.)
For each process, as initial knowledge,
an upper bound ${n'}$ on the number $n$ of processes is given.
We assume that ${n'} \leq \beta n$ holds 
for some constant $\beta \geq 1$, 
but $\beta \geq 1$ is unknown to any process.

To simulate the central daemon,
we use a randomized voting scheme so that
no two processes within two hops execute at the same time.
An enabled process selects a random number 
uniformly at random from $1,2,..., K {n'}$, 
where $K \geq 2$ is a constant, 
and an enabled process with the largest random number 
among enabled processes within two hops 
wins to execute a guarded command.

\subsection{The transformer}

First, we explain local variables of each process ${\cal P}_{i}$.
In general, 
each process $P_{i}$ of the target algorithm 
has one or more local variables.
However, 
for the sake of simplicity of explanation, 
it is assumed that each $P_{i}$ has a single local variable $x_{i}$.
The local variables of ${\cal P}_{i}$ of the transformed algorithm
include $x_{i}$ and some housekeeping variables.
The primary housekeeping variable is a cache.
Each ${\cal P}_{i}$ has 
a \emph{cache} $C_{i}[{\cal P}_{j}]$ of $x_{j}$ 
for each ${\cal P}_{j} \in N_{i}$.
Instead of reading $x_{j}$ of neighbor ${\cal P}_{j}$,
${\cal P}_{i}$ reads the cache $C_{i}[{\cal P}_{j}]$.
In case ${\cal P}_{i}$ updates the value of $x_{i}$, 
${\cal P}_{i}$ broadcasts the new value of $x_{i}$ to neighbors, 
and each neighbor ${\cal P}_{j}$ updates its cache.
In order to update the value of $x_{j}$ of some neighbor,
${\cal P}_{i}$ updates its cache for ${\cal P}_{j}$, and
${\cal P}_{i}$ broadcasts the new value of $x_{j}$ to neighbors. 
If ${\cal P}_{j}$ finds that $x_{j}$ is updated by ${\cal P}_{i}$, 
${\cal P}_{j}$ broadcasts the new value of $x_{j}$ to neighbors. 
Subsequently, each neighbor ${\cal P}_{k}$ of ${\cal P}_{j}$ updates its cache.
The correctness of cache contents is important 
to simulate the target algorithm in the message passing model.
In this paper, 
we call such a correctness \emph{cache coherency}.

\begin{definition}
\label{DEF:TR:CACHE-COHERENCY}
We say that cache is \emph{coherent}
iff, 
for each ${\cal P}_{i} \in V$, 
$C_{i}[{\cal P}_{j}] = x_{j}$ holds
for each ${\cal P}_{j} \in N_{i}$ and 
for each local variable $x_{j}$ of ${\cal P}_{j}$.
\end{definition}

The major local variables used by
the transformer at each ${\cal P}_{i}$ are as follows.
\begin{itemize}
\item $x_{i}$ 
      is to simulate the local variable of the target algorithm.
\item $C_{i}[{\cal P}_{j}]$ 
      is the cache of $x_{j}$ of ${\cal P}_{j} \in N_{i}$. 
\item $r_{i}$ 
      is a random number to 
      select a process to execute a guarded command of 
      the target algorithm.
\item $g_{i}$ 
      is true iff ${\cal P}_{i}$ is enabled.
\end{itemize}

The target algorithm 
in the R(1)W(1) model under the unfair central daemon 
is simulated in five phases.
The central daemon is simulated by 
the {distance-two} local mutual exclusion 
between processes in two hops
based on randomized voting.
That is, no two processes within two hops execute 
their guarded commands concurrently.

\begin{itemize}
\item Phase 1: 
      Each process ${\cal P}_{i}$ (locally) broadcasts the value of $x_{i}$.
      Each process ${\cal P}_{i}$ receives messages, and 
      it updates its cache $C_{i}[{\cal P}_{j}]$
      for each received message from ${\cal P}_{j} \in N_{i}$.
      Then, ${\cal P}_{i}$ computes in $g_{i}$ 
      whether some guards of the target algorithm is true or not, and
      if true, it generates a random number in $r_{i}$.

\item Phase 2: 
      If some guards of the target algorithm is true,
      ${\cal P}_{i}$ broadcasts a random number $r_{i}$.
      Subsequently, ${\cal P}_{i}$ receives messages from neighbors.
      If a process with the maximum random value is unique,
      let $w_{i}$ be the sender process ID of the maximum value.
      Here, $w_{i}$ is the winner candidate at ${\cal P}_{i}$.
      
\item Phase 3: 
      If a winner candidate is elected in the previous phase,
      ${\cal P}_{i}$ broadcasts the process ID of the winner candidate. 
      If ${\cal P}_{i}$ receives a message 
      from each neighbor ${\cal P}_{j} \in N_{i}$ and 
      ${\cal P}_{i}$ is the winner candidate at all neighbors, then
      ${\cal P}_{i}$ is the winner among processes within two hops, and
      it executes the command of the target algorithm.

\item Phase 4: 
      If ${\cal P}_{i}$ executed the command in the previous phase, 
      the local variables of ${\cal P}_{i}$ and neighbors are modified.
      ${\cal P}_{i}$ broadcasts the new values to neighbors.

\item Phase 5: 
      If the local variables of ${\cal P}_{i}$ are modified by some neighbor, 
      ${\cal P}_{i}$ broadcasts the new values to neighbors.
\end{itemize}

\begin{algorithm}
\caption{\textsf{TrR1W1} for each process ${\cal P}_{i} \in V$}
\label{TR:ALG:TRR1W1}
\KwConst\\
\Indp
$K \geq 2$ 
  \tcp*{Design parameter for randomized voting}
$R = K {n'}$
  \tcp*{Range of random numbers}
\Indm
\BlankLine
\KwLVar\\
\Indp
$x_{i}$
  \tcp*{The state of the target algorithm}
$C_{i}[{\cal P}_{j}]$
  \tcp*{Cache of $x_{j}$ for each ${\cal P}_{j} \in N^{(1)}_{i}$}
$r_{i}$
  \tcp*{Random number for probabilistic voting}
$g_{i}$
  \tcp*{Whether there is a true guard or not}
$w_{i}$
  \tcp*{Process ID with the largest vote}
$M_{i}$
  \tcp*{Message buffer}
\Indm
\BlankLine
\While{\textbf{true}}{
   \Phase(// Cache refresh \& evaluation of guards){\textnormal{\bf 1}}{
     {\XBCAST} $x_{i}$  \;
     {\XRECV}; $M_{i} :=$ messages received  \;
     Update $C_{i}$ according to $M_{i}$ \;
     $g_{i} := $ (True iff there is a true guard) \;
     \If{$g_{i}$}{
       $r_{i} := $ (select from $\{1,2,...,R\}$, u.a.r.)  \;
     }
   }
   \Phase(// Voting by random numbers){\textnormal{\bf 2}}{
     \If{$g_{i}$}{
       {\XBCAST} $r_{i}$ \;
     }
     {\XRECV};  $M_{i} :=$ messages received \;
     $w_{i} := \XNULL$\;
     \If{$(M_{i} \ne \emptyset)$ 
         and $($the maximum value among received messages is unique$)$}{
       $w_{i} := $ $($the sender process ID of the maximum value$)$ \;
     } 
   }
   \Phase(// The winner executes a command){\textnormal{\bf 3}}{
     \If{$w_{i} \ne \XNULL$}{
       {\XBCAST}\ $w_{i}$  \;
       {\XRECV}; $M_{i} :=$ messages received  \;
       \If{$($received messages from all neighbors$)$ $\land$ 
           $({\cal P}_{i}$ is the winner at each neighbor$)$}{
         Execute a command and update $x_{i}$ and $C_{i}$ \;
       }
     }
   }
   \Phase(// Value propagation to one-hop neighbors){\textnormal{\bf 4}}{
     \If{$($A command is executed in Phase 3$)$}{
       {\XBCAST} $x_{i}, C_{i}$  \;
     }
     {\XRECV}; $M_{i} :=$ messages received  \;
     Update $C_{i}$ and $x_{i}$ according to $M_{i}$  \;
   }
   \Phase(// Value propagation to two-hop neighbors){\textnormal{\bf 5}}{
     \If{$(x_{i}$ is updated in Phase 4$)$}{
       {\XBCAST} $x_{i}$  \;
     }
     {\XRECV}; $M_{i} :=$ messages received  \;
     Update $C_{i}$ according to $M_{i}$  \;
   }
}
\end{algorithm}

\subsection{Proof of correctness}

For each cycle $t \geq 1$ and each ${\cal P}_{i} \in V$,
$r_{i}(t)$ be the random value $r_{i}$ 
at the second phase of cycle $t$.

In a self-stabilizing setting,
processes may start arbitrary point of their algorithm.
That is, in the initial cycle of execution,
processes may start their execution from Phase~2 or subsequent phases. 
The next lemma is based on the assumption on reliable communication.

\begin{lemma}
\label{LEM:TR:CHACHE-COHERENCY}
After each process executes Phase~1,
the cache becomes coherent.
\end{lemma}
\begin{proof}
In Phase~1, 
each process broadcasts the value of its local variable to neighbors.
Then, each process receives the message and updates its cache.
Because it is assumed that message transmission is reliable, 
the cache becomes coherent after Phase~1.
\end{proof}

Below, 
we observe the execution of processes 
after each process executes Phase~1.
That is, we observe the second or subsequent cycles ($t \geq 2$) 
of the execution.

\begin{lemma}
\label{LEM:TR:EX-EXCLUSION}
No two processes ${\cal P}_{i} \in V$ and ${\cal P}_{j} \in N^{(1,2)}_{i}$
execute a guarded command at the same cycle.
\end{lemma}
\begin{proof}
By Lemma~\ref{LEM:TR:CHACHE-COHERENCY},
after each process executes Phase~1 once, 
the cache becomes coherent after each process receives messages
sent at the beginning of Phase~1.
Therefore, 
for each ${\cal P}_{i}$,
the value of $g_{i}$ is consistent in the sense that
$g_{i}$ is true 
iff ${\cal P}_{i}$ (the process in the target algorithm) is enabled.
Then, each process generates a random number if it is enabled, 
and processes exchange random numbers.
In case ${\cal P}_{j}$ is a neighbor of ${\cal P}_{i}$, 
${\cal P}_{i}$ and ${\cal P}_{j}$ 
do not execute a guarded command at the same time
because these each random number cannot be the maximum among neighbors.
In case ${\cal P}_{j}$ is a process in two hops from ${\cal P}_{i}$, 
there exists a process ${\cal P}_{k}$ 
such that it is a common neighbor of ${\cal P}_{i}$ and ${\cal P}_{j}$.
When ${\cal P}_{k}$ receives random number 
from ${\cal P}_{i}$ and ${\cal P}_{j}$, 
${\cal P}_{k}$ sends a process ID whose random number is uniquely the largest.
Therefore, it is not possible for
${\cal P}_{i}$ and ${\cal P}_{j}$ to be winners simultaneously. 
\qed
\end{proof}

\begin{lemma}
\label{LEM:TR:STABLE-COHERENCY}
If the cache becomes coherent, 
it remains so thereafter.
\end{lemma}
\begin{proof}
It is sufficient to show that any conflict of updates never occurs, 
that is, 
no two processes modify the same local variable and
the same cache entry concurrently.

If the number of enabled processes is at most one, 
no conflict occurs and the lemma holds clearly.

Suppose that two or more processes are enabled.
By Lemma~\ref{LEM:TR:EX-EXCLUSION}, 
after each process executes Phase~1 once,
the cache becomes coherent, and  
no two processes within two hops execute a guarded command concurrently 
thereafter.
Let ${\cal P}_{i}$ and ${\cal P}_{j}$ be any enabled processes. 
The distance between them is three or more hops.
Therefore,  
${\cal P}_{i}$ and ${\cal P}_{j}$ 
never modify the same local variable at the same time.
Furthermore, it means that
there is no cache entry which need to be updated at the same time.
The execution of the {\XBCAST} primitive in Phases~4 and 5 
results in the coherent state of the cache.
Therefore, 
once the cache coherency condition is satisfied,
it remains so forever.
\qed
\end{proof}

\begin{lemma}
\label{LEM:TR:SIMULATION}
Any execution by \textsf{TrR1W1} 
simulates the execution of the target algorithm 
in the R(1)W(1) under the unfair central daemon.
\end{lemma}
\begin{proof}
We observe the execution of \textsf{TrR1W1} 
after each process executes Phase~1 once.
Let $t \geq 1$ be the cycle number.

By Lemma~\ref{LEM:TR:EX-EXCLUSION}, 
no two processes within two hops execute a guarded command 
at the same time thereafter.
For each cycle $t \geq 2$, 
let $X(t) = \{ {\cal P}^{(t)}_{1}, {\cal P}^{(t)}_{2}, ...,
               {\cal P}^{(t)}_{|X(t)|} \}$
be the set of processes that execute a command in Phase~3 in cycle $t$.
Because the distance between any two processes in $X(t)$ is three or more, 
parallel execution of all the processes in $X(t)$ in a single step and
a serial execution of processes
${\cal P}^{(t)}_{1}, {\cal P}^{(t)}_{2}, ..., {\cal P}^{(t)}_{|X(t)|}$
in this order
result in the same local variable values and cache values.
Hence 
execution of processes in \textsf{TrR1W1} 
is equivalent to some serial execution, 
which is equivalent to the unfair central daemon.

For each cycle $t \geq 2$, 
by Lemma~\ref{LEM:TR:STABLE-COHERENCY}, 
the transformer maintains cache of local variables within the same cycle
in Phases~4 and 5,   
the composite atomicity of the R(1)W(1) model is simulated. 
\qed
\end{proof}

For each cycle $t$ and each ${\cal P}_{i} \in V$, 
let 
${H}^{(1,2)}_{i}(t) \subseteq N^{(1,2)}_{i}$
be the set of enabled processes in $N^{(1,2)}_{i}$, and
${H}(t) \subseteq V$ 
be the set of all enabled processes, 
i.e., ${H}(t) = \cup_{{\cal P}_{i} \in V} {H}^{(1,2)}_{i}(t)$.

\begin{lemma}
\label{LEM:TR:PROB-EXEC-GC}
For each cycle $t \geq 2$,
if there exists an enabled process,
the probability that at least one process executes a command 
is at least some constant probability $c > 0$.
\end{lemma}
\begin{proof}
Processes in ${H}^{(1,2)}_{i}(t)$ compete with ${\cal P}_{i}$ 
to execute their guarded commands.
If the set ${H}(t)$ is empty, 
i.e., there exists no enabled processes, 
emulation of the target algorithm is stabilized, and
each process executes {\XBCAST} only in Phase~1.
If the set size of ${H}(t)$ is 1, 
only one process is enabled and 
the process definitely executes a guarded command.
In the following, 
we assume that the set size of ${H}(t)$ is two or more, and
let ${\cal P}_{i}$ and ${\cal P}_{j}$ be any two processes in ${H}(t)$.

For any two processes ${\cal P}_{i}, {\cal P}_{j} \in {H}(t)$, 
the probability of an event that they generate different random numbers is
$1 - 1/R$.
For any ${\cal P}_{i} \in {H}(t)$, 
the probability of an event that
the random number 
$r_{i}(t)$ is different from 
$r_{j}(t)$ for each ${\cal P}_{j} \in {H}^{(1,2)}_{i}(t)$ is 
\begin{align*}
        \biggl( 1 - \frac{1}{R} \biggl)^{|{H}^{(1,2)}_{i}(t)|} 
& 
\geq    \biggl( 1 - \frac{1}{R} \biggl)^{n-1} 
\geq    \biggl( 1 - \frac{1}{K n} \biggl)^{n-1} 
>       \XNAPIERNUM^{-1/K} .
\end{align*}

And given that this holds,
the probability of an event that 
$r_{i}(t)$ is larger than any $r_{j}(t)$ 
for each ${\cal P}_{j} \in {H}^{(1,2)}_{i}(t)$ is 
$1/(|{H}^{(1,2)}_{i}(t)| + 1) \geq 1/|{H}(t)|$ 
due to the symmetry of processes.
The probability of an event that 
$r_{i}(t)$ is larger than 
any $r_{j}(t)$ for each ${\cal P}_{j} \in {H}^{(1,2)}_{i}(t)$ is at least
\begin{align*}
     \frac{1}{|{H}(t)|} \cdot \XNAPIERNUM^{-1/K} 
 =    \frac{1}{S |{H}(t)|}, 
\end{align*}
where $S = \XNAPIERNUM^{1/K}$.
Because $K \geq 2$,
we have $1 < S \leq \XNAPIERNUM^{1/2} = 1.64872\cdots$.

The probability $I_{n}$ of an event that
there exists at least one process, say ${\cal P}_{i} \in {H}(t)$, 
such that $r_{i}(t)$ is 
larger than any $r_{j}$ for each ${\cal P}_{j} \in {H}^{(1,2)}_{i}(t)$ is 
\begin{align*}
I_{n}
& \geq   1 - \prod_{{\cal P}_{i} \in {H}(t)} \biggl( 1 - \frac{1}{S |H(t)|} \biggr) 
  =      1 - \biggl( 1 - \frac{1}{S |H(t)|} \biggr)^{|H(t)|}  
\\
& \geq     1 - \XNAPIERNUM^{-1/S}  
  \geq     1 - \XNAPIERNUM^{-1/\XNAPIERNUM^{1/2}} 
  \approx  0.45476 .
\end{align*}

Hence, if there exists an enabled process, 
at least one process executes a guarded command
with probability at least $c \approx 0.45476$.
\qed
\end{proof}

Finally, we have the following theorem.

\begin{theorem}
\label{THR:TR:CONVTIME}
Let $A$ be a silent self-stabilizing algorithm 
in the R(1)W(1) model 
that stabilizes in $T_{A}$ moves in the worst case
under the unfair central daemon.
Let $A'$ be the 
transformed algorithm of $A$ by \textsf{TrR1W1}.
Then, $A'$ is a self-stabilizing algorithm 
in the synchronous message passing model
that stabilizes in $O(T_{A})$ expected rounds.
\end{theorem}
\begin{proof}
By Lemmas~\ref{LEM:TR:CHACHE-COHERENCY} and \ref{LEM:TR:STABLE-COHERENCY}, 
at the beginning of Phase~1 in the second cycle $t = 2$, 
the cache is coherent and it remains so thereafter.
By Lemmas~\ref{LEM:TR:EX-EXCLUSION} and \ref{LEM:TR:SIMULATION},
if some process in $A'$ executes a guarded command  
then there exists an equivalent serial execution in $A$
in each cycle $t \geq 2$. 
Hence, for any execution of $A'$, 
there exists an equivalent serial execution in $A$. 
Because $A$ is self-stabilizing under the unfair central daemon, 
any execution of $A'$ converges to
some configuration which corresponds to a legitimate configuration of $A$.

By Lemma~\ref{LEM:TR:PROB-EXEC-GC},
for each cycle $t \geq 2$, 
at least one process executes a guarded command 
with probability at least some constant $c > 0$. 
Let $\tau_{A}$ 
be the worst case convergence time of algorithm $A$.
If the number of moves is $\tau_{A}$,  
the execution of $A'$ converges to
some configuration which corresponds to a legitimate configuration of $A$.
If we execute the transformed algorithm for 
$\tau_{A} / c$ cycles (or, equivalently, $5 \tau_{A} / c$ rounds),
the expected number of moves is at least $\tau_{A}$.  
\qed
\end{proof}

Let us we evaluate the overhead factor of 
our transformation by \textsf{TrR1W1}
in terms of message complexity.

\begin{theorem}
\label{LEM:TR:PERFOMANCE}
Let $A$ be the target algorithm in the R(1)W(1) model, and
$A'$ be the transformed algorithm of $A$.
Let $T_{A}$ be the maximum number of moves for convergence of $A$.
The expected total number of executions of {\XBCAST} of 
the transformed algorithm is $O(n T_{A})$, 
where 
$n$ is the number of processes.
\end{theorem}
\begin{proof}
By Theorem~\ref{THR:TR:CONVTIME},
if the transform algorithm $A'$ is executed 
for $T_{A} / c$ cycles (or equivalently, for $5 T_{A} / c$ phases),
where $c$ is the lower bound of the probability 
shown in the proof of Lemma~\ref{LEM:TR:PROB-EXEC-GC},
the expected number of processes which executes a guarded command is at least $T_{A}$.
In each phase, 
every process may broadcast a message by {\XBCAST}.
Hence the expected total number of invocations of {\XBCAST} is 
$5 n T_{A} / c = O(n T_{A})$.
\qed
\end{proof}

After the transformed target algorithm converges, 
any message loss does not break the coherency of cache, and
configuration remains legitimate. 
That is assumption of the reliability of communication 
is needed during convergence.
This owes to the assumption that the target algorithm is silent.

\section{Conclusion}
\label{SEC:CONCLUSION}

In this paper, we proposed a new communication model, the R(1)W(1) model, 
which allows each process atomically update local variables 
of neighbor processes. 
We propose some self-stabilizing distributed algorithms
in the R(1)W(1) model.
We also proposed {an example} transformer \textsf{TrR1W1} 
to run such algorithms
in the synchronous message passing model, and
showed that the expected overhead of transformation is 
$O(1)$ in time complexity and
$O(n)$ in message complexity.
The design of a transformer is independent from the R(1)W(1) model, and
development of an efficient transformer is a future task.

We mentioned that the R(1)W(1) model is further generalized to 
the R($d_{r}$)W($d_{w}$) model, 
where $d_{r} \geq 1, d_{w} \geq 0$.
Developing an efficient transformer for the R($d_{r}$)W($d_{w}$) model
is a future work.

%


\end{document}